\documentclass[12pt]{article}
\usepackage{amsthm}
\newcommand{\bfm}[1]{\mbox{\boldmath{$#1$}}}
\usepackage{amssymb,amsfonts,amsmath}

\usepackage{graphicx}

\newcommand{\beq}{\begin{eqnarray}}
\newcommand{\eeq}{\end{eqnarray}}
\newcommand{\beqs}{\begin{eqnarray*}}
\newcommand{\eeqs}{\end{eqnarray*}}

\usepackage{graphics,epsfig}
\usepackage{color}

 \newtheorem{theorem}{Theorem}
\newtheorem{corollary}{Corollary}
\newtheorem{definition}{Definition}

\begin{document}

\title{Disassociation energies for the finite density $N$-body problem}
\author{D.J. Scheeres $<$scheeres@colorado.edu$>$ \\Smead Department of Aerospace Engineering Sciences\\University of Colorado Boulder}
\date{\today}
\maketitle

\begin{abstract}
This paper considers the energy required for collections of finite density bodies to undergo escape under internal gravitational interactions alone. As the level of the system energy is increased there are different combinations of components that can escape, until the total energy becomes positive, when the entire system can undergo mutual disruption. The results are also defined for bodies modeled as a continuum. These results provide rigorous constraints for the disruption of rubble pile asteroids when only considering gravitational interaction effects, with the energy provided by rotation of an initial collection of the system. These issues are considered for discrete particles in the $N$ body problem and for size distributions of discrete particles in the continuum limit.
\end{abstract}

\section{Introduction}

This paper explores the energy required for a finite density $N$-body problem to undergo self-disruption under its internal dynamics alone, meaning that some components of the body can escape to infinity without additional energy being added to the system. This problem has direct application to the evolution of rubble pile asteroids in the solar system, where there is ample evidence that these bodies may be slowly spun up over time until they have sufficient energy to fission, with some of the components eventually escaping from the system and forming asteroid pairs or clusters \cite{pravec_fission, Pravec_clusters}. The process of self-disruption and the limits for the escape of components has been studied previously for specific body types and scenarios \cite{scheeres_F2BP, scheeres_minE, F4BP_chapter, F3BP_scheeres}. The current work unifies some aspects of these other studies and extends these results to size distributions, accounts for porosity within a rubble pile body and respects the indivisible nature of individual ``boulders'' that may exist in rubble pile bodies. 

The main purpose of this paper is to provide a clear description of how a finite density $N$ body problem can be disaggregated and what the necessary energy levels are for such disaggregation to occur. The results can be used to build up a rigorous description for the formation of asteroid pairs and clusters, and to explore the limits on energy requirements for disruption. 

The paper is organized as follows. First, the basic model is presented for the finite density, or full, $N$ body problem and  specific terminology is developed for how different sub-sets of the Full $N$ body problem are defined. Next, the basic mechanics that control the evolution of these systems is reviewed. Following this, the energetics of the problem are discussed identifying the free energy, mutual potential energy, self-potential energy, cohesive energy and disruption energy. In addition the minimum energy configurations for a given subset of the Full $N$ body problem are defined. With these terms defined, Hill Stability is defined in the current context and a few theorems that relate the total energy of a system and the Hill Stability of certain sub-components of the problem are proven. Finally a few examples of how these theorems can be applied to develop energy constraints on the disruption of a full $N$ body problem for both discrete and continuous bodies is given.  

\section{Model}

The finite-density $N$-body problem is posed. It is assumed that all of the component bodies are rigid and can thus rest on each other to form new varieties of equilibrium \cite{scheeres_minE}.{\color{black} We also invoke the continuum hypothesis which enables us to treat these bodies as infinitely divisible and, thus, to be able to deform into new shapes. We do not develop a detailed model of this reshaping process, however, and only invoke it to identify minimum energy configurations}.

\subsection{Degrees of Freedom}
A system of $N$ finite density bodies with masses $m_i$ and inertia dyads $\bfm{I}_i$, $i = 1,2,\ldots, N$ will have $6N$ degrees of freedom in general, with each body having 3 translational and 3 rotational coordinates, represented as $\bfm{r}_i$ and $\bfm{T}_i$, $i = 1, 2, \ldots, N$. For translational degrees of freedom the position vector $\bfm{r}_i$ is the body center of mass relative to a system barycentric frame. For the rotational degrees of freedom the rotation dyad $\bfm{T}_i$ specifies the body orientation relative to an inertial frame. The individual entries of $\bfm{T}_i$ are the dot products between the body and inertial frame, with the dyad defined as taking a vector in the body frame and transforming it into the inertial frame. As the inertia dyads are generally stated in the body-frame, they can be mapped into the inertial frame by $\bfm{T}_i\cdot\bfm{I}_i\cdot\bfm{T}_i^T$. The rotation dyad has the constraint that $\bfm{T}_i^T\cdot\bfm{T}_i = \bfm{T}_i\cdot\bfm{T}_i^T = \bfm{U}$, where $\bfm{U}$ is the identity dyad. 
%While the rotation dyad has a bit more overhead associated with it as compared to a minimal set of Euler angles, ultimately it is more convenient to work with. 

Of more interest than the inertial locations and orientations are the relative positions and orientations between the objects, defined as $\bfm{r}_{ij} = \bfm{r}_j - \bfm{r}_i$ and $\bfm{T}_{ij} = \bfm{T}^T_j \cdot\bfm{T}_i$, $1 \le i < j \le N$. The total number of degrees of freedom can be reduced through translational and angular momentum symmetries, however in this paper such reductions are not made. 

With the finite density assumption, any two bodies have a constraint on how close their centers of mass can come to each other, denoted as $r_{ij} \ge d(\bfm{r}_{ij}, \bfm{T}_{ij})$, where equality occurs when the two bodies touch. This form of the constraint implicitly implies that the bodies are mutually convex, although this assumption can be relaxed.  
In the simplest model, the individual bodies are spheres with a given size and density, however the problem is easily generalized to arbitrarily shaped components, that exert torques on each other and whose rotational dynamics must also be accounted for \cite{scheeres_minE}. The spherical model is frequently used to compute specific examples, however the general results developed are largely independent of such an assumption. 

If a given rigid body is sub-divided into components, then the system gains additional degrees of freedom that, ideally, must be specified. When considering divided bodies we will generally assume that they take on a spherical shape (as will be justified and explained later), meaning that the relative positions of the components are the most important degrees of freedom.

{\color{black}
\subsection{Partitions, Collections and Components}
Our collection of $N$ finite density bodies can be represented generically as the collection ${\cal B} = \left\{{\cal B}_1, {\cal B}_2, \ldots, {\cal B}_N\right\}$. Each of the bodies ${\cal B}_i$ has its own unique mass, $m_i$, as well as its shape, density and other unique specifiers. The total mass of the system is $m = \sum_{i=1}^N m_i$. 

Each of these bodies can also be arbitrarily split into a countable set of components, or ${\cal B}_i = \left\{{\cal B}_{i_1}, {\cal B}_{i_2}, \ldots, {\cal B}_{i_j},\ldots, {\cal B}_{i_{N_i}} \right\}$ where $1\le N_i < \infty$. The mass of the $j$th component is specified as $m_{i_j} = m_i \mu_{i_j}$ where $\mu_{i_j}$ is its mass fraction, defined for $j = 1,2,\ldots,N_i$ such that $0<\mu_{i_j}\le 1$ and  $\sum_{j=1}^{N_i} \mu_{i_j} = 1$. 

A single collection of bodies can then be uniquely specified by a set of integers that specify individual bodies and sub-components of individual bodies. For example, $\mathbb{I} = \left\{ 1, 2, 3_1, 3_2, \ldots i_j \right\}$ where $i \in [1,N]$ and $j \in [1,N_i]$. Then the collection of bodies is specified as ${\cal B}({\mathbb{I}}) = \left\{ {\cal B}_{i_j} | i_j \in \mathbb{I} \right\}$. In terms of notation, if $N_i = 1$ the index $j$ is not specified, and if $N = 1$ the index $i$ is not specified. 

Given this notation, we define the global partition of a mass distribution ${\cal B}$. 
\begin{definition}
\label{def:partition}
The {\bf Global Partition} of a mass distribution ${\cal B}$ is defined by the set of indices ${\cal I}_M = \left\{ \mathbb{I}_1, \mathbb{I}_2, \mathbb{I}_3, \ldots, \mathbb{I}_i, \ldots, \mathbb{I}_M \right\}$ and is ${\cal B}({\cal I}) = \left\{ {\cal B}(\mathbb{I}_i) | \mathbb{I}_i \in {\cal I}\right\}$. The $M$ subscript will be used when we wish to emphasize the number of components, but will be suppressed in other situations. 
Such a partition is called a global partition of the mass distribution if ${\cal B}(\mathbb{I}_i) \ne \emptyset$, ${\cal B}(\mathbb{I}_i) \bigcap {\cal B}(\mathbb{I}_j) = \emptyset$ for $i\ne j$, and $\bigcup_{\mathbb{I}_i\in{\cal I}} {\cal B}(\mathbb{I}_i) = {\cal B}$. In words, ${\cal B}({\cal I})$ is a global partition if its individual elements ${\cal B}(\mathbb{I}_i)$ are ordered, non-empty, disjoint and if their union consists of all the mass in the system. 
\end{definition}

%\begin{definition}
%{\bf Global Partition:}
%A global partition ${\cal I}$ of the total mass distribution ${\cal B}$, designated as ${\cal B}_{\cal I} \left\{ \bfm{I}_1, \bfm{I}_2, \bfm{I}_3, \ldots, \bfm{I}_i, \ldots, \bfm{I}_M \right\}$, is an ordered collection of non-empty, disjoint subsets ${\cal B}_{I_i}$, $i = 1, 2, 3, \ldots, M$ where $I_i \in {\cal I}$, whose union equals ${\cal B}$. Symbolically, 
%\end{definition}
%}

If the bodies are discrete, uniquely identifiable and not subdivided the number of possible ways they can be split into different unique configurations is equal to the Bell Number, $B_N$. The number of possible partitions increases quite rapidly with $N$, with the first members of the sequence equal to $B_0 = 1$, $B_1 = 1$, $B_2 = 2$, $B_3 = 5$, $B_4 = 15$, $B_5 = 52$, and with the upper bound $B_N < \left[\frac{0.792 \ N}{\ln(N+1)}\right]^N$ \cite{bell_number}. 
If the bodies are discrete but not unique, e.g., if all of the bodies are identical and cannot be distinguished, the number of possible combinations is reduced. Then, given $N$ bodies, the number of different groupings the collection can be separated into is equal to its integer partition. The integer partitions of a number $N$ are the different ways the number $N$ can be expressed as a sum of positive integers, and the total number of ways it can be so formed is designated as $p(N)$. For example, $p(3) = 3$, with partitions 3, 1+2 and 1+1+1. Figure \ref{fig:configurations} graphically shows the integer partitions for $N=2 \rightarrow 8$ using a Young diagram. The additional information in the figure is defined later. For large $N$ an asymptotic formula for the number of integer partitions is given as $p(N) \sim \exp (\pi \sqrt{2N/3})$ \cite{andrews1998theory}, and in this case the number of possible partitions becomes exponentially large. 
Finally, if the bodies can be sub-divided the number of ways the system can be split up is uncountable, even if every possible partition consists of a countable set of components. 
}

\begin{figure}[ht!]
\centering
\includegraphics[scale = 0.35]{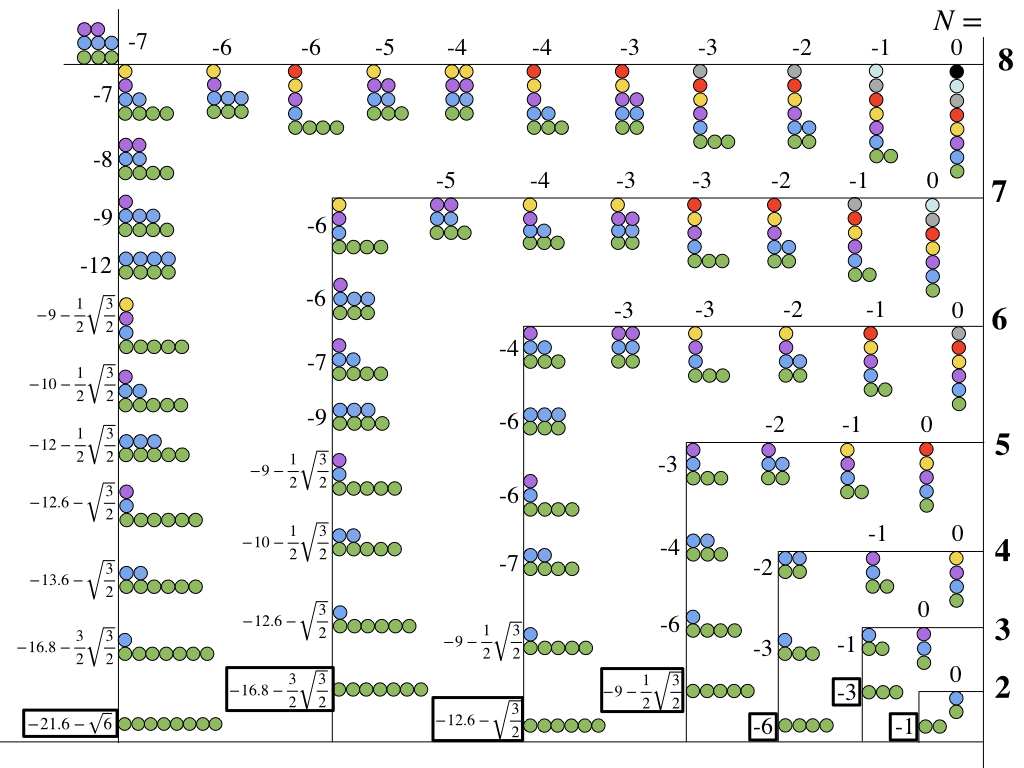}
\caption{Young Diagram for $N = 2 \rightarrow 8$, showing the minimum energy for each partition, scaled by ${\cal G}m^2/d$. Per Corrolary 2, the minimum disassociation energy of any partition is equal to the minimum energy of the next partition in the diagram. The boxed numbers are minimum energies for that number of bodies, while the other numbers are the disassociation energies for those partitions. }
\label{fig:configurations}
\end{figure}

\subsection{Configurations}

In addition to dividing bodies into different partitions, it is also important to specify their relative positions and orientations within each sub-partition. For finite density bodies this means specifying their relative positions and relative attitudes, denoted as $\bfm{r}_{ij}$ for the relative position vector between two bodies and $\bfm{T}_{ij}$ for the relative orientation dyad between two bodies. 
\begin{definition}
\label{def:configuration}
{\bf Configurations:} 
A {\bf Configuration} is defined as a sub-partition $\mathbb{I}$ and the set $\bfm{Q}({\mathbb{I}}) = \left\{ \bfm{r}_{ij}, \bfm{T}_{ij} | i,j\in \mathbb{I}\right\}$, and defines the relative orientation and position of components within the sub-partition $\mathbb{I}$. 
%A {\bf Configuration} is then denoted as $\bfm{C}_i = \left\{ \bfm{I}_i, \bfm{Q}_i \right\}$ where the short-hand notation $i$ is used to match the sub-partition to the configuration. 
%\end{definition}
%
%\begin{definition}
%{\bf Full Configuration:} 
The {\bf Global Configuration} of the system is designated as the combination of a global partition ${\cal I}$ with the associated entire configuration of the sub-partitions ${\cal Q}({\cal I}) = \left\{ \bfm{Q}(\mathbb{I}_1), \bfm{Q}(\mathbb{I}_2), \ldots, \bfm{Q}(\mathbb{I}_M)\right\}$. 
%The entire partition plus configuration is denoted as ${\cal C} = \left\{ {\cal I}(M,N), {\cal Q}(N,M)\right\}$. 
\end{definition}
What is left unsaid is what the relative configurations are between two sub-partitions. This could be specified in general, but as this paper is only interested in when sub-partitions have escaped relative to each other, meaning that all of the relative positions between two bodies in different sub-partitions will be infinite. 

\section{Mechanics}

Having defined partitions and configurations, the gravitational interactions between them and the resulting mechanics must be summarized. 

\subsection{The Amended Potential and Constraints} 

Key to the discussion is the amended potential for the finite density $N$-body problem. This has been defined in more detail elsewhere \cite{scheeres_minE}, and thus is stated without derivation as 
\beq
	{\cal E} & = & \frac{H^2}{2 \ I_H} + {\cal U}
\eeq
where $H$ is the total angular momentum magnitude of the entire system, $I_H$ is the total moment of inertia of the entire system about the angular momentum direction, and ${\cal U}$ is the total gravitational potential energy of the system. The amended potential does not define a constant of motion for the system, but has a specific relationship to the constant energy of the system which can be derived from Sundman's Inequality \cite{scheeres_minE} 
\beq
	{\cal E} & \le & E
\eeq
where $E = T + {\cal U}$ is the conserved combination of kinetic, $T$, and total gravitational potential, ${\cal U}$, and is a constant. It is important to note that for a fixed level of angular momentum the amended potential is only a function of the system configuration degrees of freedom, and is not a function of the system velocities.  In general the amended and the total energies are only equal when the system is in a relative equilibrium configuration or is momentarily at a zero-velocity condition. 

Define the general potential energy as 
\beq
	{\cal U} & = & \sum_{1\le i < j \le N} {\cal U}_{ij}(\bfm{r}_{ij}, \bfm{T}_{ij}) + \sum_{i=1}^N {\cal U}_{ii}
\eeq
where ${\cal U}_{ij}$ is the mutual potential between two bodies and ${\cal U}_{ii}$ is the self potential of the $i$th body. 
When dealing with rigid bodies the mutual and self potential energies are easily tracked separately. For continuum models, however, these two are often mixed and must be separated. In the following the total self-potential of the system is refered to as ${\cal U}_{Self} = \sum_{i=1}^N {\cal U}_{ii}$. The mutual potential is defined as ${\cal U}_{Mutual} =  \sum_{1\le i < j \le N} {\cal U}_{ij}(\bfm{r}_{ij}, \bfm{T}_{ij})$ and is a function of the system configuration and thus can have changing values. 

The total moment of inertia about $H$ is 
\beq
	I_H & = & \hat{\bf H} \cdot \left[ \sum_{1\le i \le N} \bfm{T}_i\cdot\bfm{I}_i\cdot\bfm{T}_i^T + \sum_{1\le i < j \le N} \frac{m_i m_j}{m_i+m_j} \left( r_{ij}^2 \bfm{U} - \bfm{r}_{ij}\bfm{r}_{ij}\right)
	\right] \cdot \hat{\bf H} 
\eeq
where $\bfm{I}_i$ is the inertia dyad of each body, $m_i$ is the mass of each body and $\hat{\bfm{H}}$ is the unit vector that points along the conserved total angular momentum of the system. 

Given these definitions one can note that ${\cal E}$ is only a function of the system configuration ${\cal Q}$ and the total angular momentum $H$, and is bounded by the total energy. Thus, this function relates the system configuration to the fundamental conserved quantities of the system. 

\subsection{Relative Equilibrium Configurations}

The starting point for this analysis is the idea that a collection of rigid bodies have settled into a relative equilibrium at a fixed value of angular momentum, meaning that all components are stationary and rotating at a constant rate. If the further restriction that the configuration be stable is imposed, then the collection has no pathway to a disrupted or deformed state. The existence of such stable states has been studied for asteroidal and cometary bodies, and the detailed conditions worked out for two, three, four and more bodies resting on each other \cite{scheeres_F2BP, scheeres_icarus_fission, scheeres_minE, F4BP_chapter, F3BP_scheeres}. 

When such a system is in relative equilibrium its amended potential is at a stationary value. 
The condition for a system to be in such an equilibrium is then that $\delta{\cal E}(\bfm{Q}^*) \ge 0$ over all allowable configuration variations in the neighborhood of $\bfm{Q}^*$ (the equality is for degrees of freedom that are not constrained while the inequality are for active constraints). Conditions for the stability of the system are $\delta^2{\cal E}(\bfm{Q}^*) > 0$, where the second variation must be positive definite in all of the unconstrained degrees of freedom \cite{scheeres_minE}. The equilibrium configuration can have all of the bodies resting on each other, or may have bodies in orbit about each other \cite{moeckel2018counting}. However, to be stable the system must be collected into one or two bodies only \cite{moeckel2017minimal}. When in such an equilibrium the total system energy is equal to the amended potential, $E = {\cal E}({\bf Q}^*)$. If the system is excited away from an equilibrium condition then instead the inequality ${\cal E}({\bf Q}^*+\delta\bfm{Q}) \le E$ holds.

\subsection{Disruption Conditions}

For a given configuration in an equilibrium, if the angular momentum $H$ is changed over time (due to exogenous effects such as the YORP effect \cite{YORP_AIV}) then the amended potential and hence total energy also increases due to the $H^2 / (2 I_H)$ term. There exist values of angular momentum where the stability of a given equilibrium can change, or more drastically an equilibrium can disappear, meaning that the system suddenly cannot remain in a relative stationary state and starts undergoing dynamical motion. In the ensuing motion, depending on the level of energy in the system and on the mass distribution of the system, the newly dynamic components of the system may escape from each other and form the aforementioned asteroid pairs or clusters. 
The thresholds for such transitions have been studied for a variety of few-body systems \cite{scheeres_F2BP, scheeres_icarus_fission, scheeres_minE,F4BP_chapter,F3BP_scheeres}. 

When a system undergoes such a transition and some components of it escape, the total angular momentum and energy across this transition ideally stay constant. However, as at least two components will have a distance $r_{ij} \rightarrow\infty$, this means that $I_H \rightarrow\infty$ and thus $H^2 / (2 I_H) \rightarrow 0$. Also, the mutual potential energy between those two components also goes to zero, or ${\cal U}_{ij}\rightarrow 0$. This leads to a transition in the system's amended potential such that 
\beq
	{\cal E} & \rightarrow & \sum_{i=1}^{M} \sum_{j,k\in\mathbb{I}_i} {\cal U}_{jk}({\bfm{Q}(\mathbb{I}_i}))
\eeq
%where $\bfm{I}_i$ denotes the $M$ sets of  bodies that remain clumped together and ${\cal U}({\bfm{I}_i})$ denotes the mutual potential within each clump. 
The system energy can then be written as $ E = T + \sum_{i=1}^{M} {\cal U}({\bfm{Q}_i}) \ge {\cal E}$, where the shorthand notation ${\cal U}({\bfm{Q}_i}) = \sum_{j,k\in\mathbb{I}_i} {\cal U}_{jk}({\bfm{Q}(\mathbb{I}_i}))$ is introduced. %\ge T + \sum_{i} {\cal U}_{min}({\bfm{I}_i})$. 
If the mutual potentials of the different subcomponents of the system can be delineated, this relation gives a test for whether such an escape can occur at the given level of energy. 

In general the energy can be repartitioned into individual amended potentials for each sub-partition, which are denoted as ${\cal E}_{\mathbb{I}_i}$ and which include the angular momentum of that sub-component, $H_i$, and its kinetic energy, $T_i$. The total angular momentum and energy budget must also include the translational kinetic energy and orbital angular momentum of the bodies escaping from each other. This paper does not probe what these different energies could be, and only focuses on the strict and rigorous limits for such escape to occur. 

\section{Hill Stability Definitions and Theorems}

The heart of this paper is to quantify the partitions that the original problem has been split into, the mutual potentials that can exist within each partition and the level of energy of the original $N$ body assemblage that will allow the body to disaggregate in this way.  

With the appropriate background given, the concept of Hill Stability and associated theorems of interest can be introduced. 
\begin{definition}
\label{def:hill}
A sub-partition $\bfm{Q}(\mathbb{I})$ of our system is {\bf Hill Stable} if there exists a positive constant $C$ such that $r_{ij} < C < \infty$ for all $i,j \in \mathbb{I}$ and for all time. 
\end{definition}

%First we define what we mean by a configuration, and then we introduce a modified concept we call Configuration Hill Stability.
%\begin{definition}
%A configuration of the $N$-body problem, denoted as $\mathcal{C}_j$, is defined as a unique grouping of the $N$ bodies into $P$ groups, with each group having a set number of bodies $q_1 \ge q_2 \ldots \ge q_P \ge 1$ and $\sum_{i=1}^P q_i = N$. We denote a given grouping as $\mathcal{C}_j = \left\{ N, P, Q = (q_1, q_2, \ldots, q_P) \right\}$. We note that there can be several different configurations as a function of $N$, or $j = 1, 2, \ldots, M$, where each will either have a different $P$ or different $Q$. 
%\end{definition}
%For example, every group of $N$ bodies has a single, unique grouping into $N$ groups, with $P = N$ and $q_i = 1$. Each group of $N$ bodies also has a single, unique grouping into 1 group, with $P = 1$ and $q_1 = N$. 
%Figure \ref{fig:configurations} shows all possible configurations for $N = 2, 3, 4, 5$. 

\begin{definition}
\label{def:hill_config}
A global partition ${\cal Q}({\cal I}_M)$ composed of $M$ sub-partitions is termed to be {\bf Configuration Hill Unstable} if each sub-partition is Hill Stable itself, but each pair of sub-partitions escape relative to each other. If we define the distance between the sub-partitions to be $d_{ij}$, then there is a sequence of times such that the mutual distances between all of the sub-partitions become infinite, even while the individual sub-partitions keep all of their components as bounded.  Given this definition, we say that the converse of this is defined to be {\bf Configuration Hill Stable}, meaning that there is at least one pair of sub-partitions that do not escape. 
\end{definition}
%These definitions also apply to continuous mass distributions, as defined by ${\cal I}$. 

The relative configuration within a sub-partition is quite important as it drives what the available gravitational potential energy of that sub-partition is. Of greatest interest for the current study is the minimum energy gravitational potential of a given partition and associated configuration. 
\begin{definition}
\label{def:minE}
{\bf Minimum Energy Configurations:} \\
The {\bf Discrete Minimum Energy Configuration} for a given sub-partition $\mathbb{I}$ is defined as 
\[
\bfm{Q}^*(\mathbb{I}) = \mbox{argmin}_{\bfm{Q}(\mathbb{I})} \ \sum_{i,j\in\mathbb{I}} {\cal U}_{ij}(Q_{ij})
\] 
and the associated minimum energy value associated with this state is 
\[
{\cal U}_{min}(\mathbb{I}) = \sum_{i,j\in\mathbb{I},Q_{ij}\in\bfm{Q}^*} {\cal U}_{ij}(Q_{ij}) + \sum_{i\in\mathbb{I}} {\cal U}_{ii}
\]

Thus the minimum energy configuration for a full partition is ${\cal Q}_{min}({\cal I})$ and the potential energy value associated with it ${\cal U}_{min}({\cal I})$. 
%\end{definition}
For the discrete minimum energy configurations to be defined it is necessary that the components have a finite distance constraint between each body, or $r_{ij} > 0$. Otherwise, the minimum potential is unbounded from below and goes to $-\infty$. We note that a minimum energy configuration exists for every collection of finite density bodies. 

For continuous bodies the minimum energy configuration for a single body can be stated precisely and is classically known to be a spherical distribution. 
%\begin{definition}
%{\bf Continuous Minimum Energy Configuration:}
The {\bf Continuous Minimum Energy Configuration} for a single body of mass $m$ and (constant) density $\rho$ is for the body to gather itself into a sphere of radius $R = \left[ \frac{3 m}{4\pi\rho}\right]^{1/3}$. Its minimum gravitational potential energy is then
\[
{\cal U}_{min}(m) = - \frac{3{\cal G}m^2}{5 R} = - \frac{3{\cal G}}{5 }\left(\frac{4\pi\rho}{3}\right)^2 R^5
\]
This result holds whether the body is a distribution of grains (under the continuum limit), or if it is a single spherical grain. Also, if the body has some angular momentum, it can be shown that the sphere is no longer the minimum energy configuration -- although with the additional angular momentum the total energy will always be greater than that of just a spherical mass distribution with no rotation.  
\end{definition}

With these definitions it is possible to establish several rigorous results on Hill Stability for the Full $N$-Body Problem. 
\begin{theorem}
\label{thm:HS}
A global partition of the Full $N$-Body Problem, $\mathcal{I}$, is Configuration Hill Stable if $E < \sum_{1\le i \le M} \mathcal{U}_{min}(\mathbb{I}_i)$. 
\end{theorem}
\begin{proof}
Assume that $E < \sum_{1\le i \le M} \mathcal{U}_{min}(\mathbb{I}_i)$ but that the system is not Configuration Hill Stable for these components. Thus one can assume that all of these components mutually escape each other. As the distances between these components go to $\infty$, $I_H \rightarrow \infty$ and $\mathcal{U} \rightarrow \sum_{1\le i \le M} \mathcal{U}_{}(\mathbb{I}_i)$. If, in addition, any of the individual components $\mathbb{I}_i$ also undergoes mutual escape, the amended potential loses that negative contribution and takes on an even larger overall value, however we do not need to consider that case here.  Applying the energy constraint then gives $\mathcal{E} = \sum_{1\le i \le M} \mathcal{U}(\mathbb{I}_i) \le E < \sum_{1\le i \le M} \mathcal{U}_{min}(\mathbb{I}_i)$. However, as $\mathcal{U}_{min} \le \mathcal{U}$, this is a contradiction, meaning that the given global partition is Configuration Hill Stable. 
\end{proof}
In this context, Configuration Hill Stability means that not all of the sub-components of the partition can mutually escape each other. If instead only some components escape, but others are bound to each other, this is equivalent to a partition with fewer components being Hill Unstable, and corresponds to a partition at a lower value of energy being Hill Unstable. 

For system components that include bodies which are subdivided we can develop a more explicit definition for Hill Stability. 
\begin{corollary}
\label{cor:cont} 
Given a system of mass $m$, its minimum energy potential is ${\cal U}_{min}(m)$. Let it be divided into a global partition ${\cal I}=\left\{ \mathbb{I}_1, \mathbb{I}_2, \ldots \right\}$ where each component $\mathbb{I}_i$ has a mass fraction $\mu_i$ and $\sum_{i}\mu_i = 1$.  Then this partition is Configuration Hill Stable if the system energy $E < {\cal U}_{min} \left[ \sum_{i} \mu_i^{5/3}\right]$. 
\end{corollary}
\begin{proof}
First note that for a constant density spherical body of radius $R$, the spherical radius of the mass fraction $\mu$ of this body is $r = \mu^{1/3}R$. Then, as the minimum potential of this mass fraction is to assemble itself into a sphere of radius $r$, then its self potential is just scaled as $r^{5} = \mu^{5/3} R^5$ and the minimum potential energy of the mass fraction is ${\cal U}(m) \mu^{5/3}$.
Then, from Theorem 1 we can show that the bounding energy $\sum_{1\le i \le M} \mathcal{U}_{min}(\mathbb{I}_i) = {\cal U}_{min}(m)  \left[ \sum_{i} \mu_i^{5/3}\right]$. Finally, as $\mu_i < 1$ and $\sum_i \mu_i = 1$ it can be shown that the summation $\sum_{i} \mu_i^{5/3}$ converges, by noting that $\mu_i^{5/3} < \mu_i$. 
\end{proof}

For spherical, discrete particles we can prove the following corollary which defines the minimum energy for disruption. 
\begin{corollary}
\label{cor:disc} 
For a discrete system of spherical bodies with a common density, let $m_{m} = \min_{i\in \mathbb{I}} m_i$ be the minimum mass of this full $N$-body system and $i_m$ be its index. Let $\mathbb{I}-i_m$ be the $N-1$ body system that excludes this minimum mass. Then, the original full $N$-body system is Hill Stable if $E < {\mathcal{U}_{min}(\mathbb{I}-i_m)}$. 
\end{corollary}
\begin{proof}
We first assert that ${\mathcal{U}_{min}(\mathbb{I}-i_m)} < {\mathcal{U}_{min}(\mathbb{I}-j)}$ where $i_m \ne j$. As the particles all have the same density, it is clear that adding an additional mass to an existing collection will decrease its minimum energy, but that it will be decreased by a larger amount for an added mass $m_j > m_{i_m}$. Thus, it stands to reason that removing a smaller mass from a collection will increase the minimum energy by a smaller amount than removing a larger mass. If we allow for different densities within the system, then this computation and comparison becomes more difficult. 

Given this, then assume that $E < {\mathcal{U}_{min}(\mathbb{I}-i_m)}$ but that the system is not Hill Stable, and escape of at least one particle can occur. If body $i_m$ escapes, $I_H \rightarrow \infty$ and $\mathcal{U}(\mathbb{I}) \rightarrow \mathcal{U}(\mathbb{I}-i_m)$, as ${\cal U}_{j i_m}\rightarrow 0$. However, Sundman's Inequality still holds, and after the escape gives $\mathcal{E} = \mathcal{U}(\mathbb{I}-i_m) \le E < \mathcal{U}_{min}(\mathbb{I}-i_m)$. However, this cannot be true by definition, and thus the contention that the system is not Hill Stable is untrue. Applying the converse, we see that the system is Hill Stable. 
\end{proof}

%It is also important to extend these Hill Stability definitions to continuous mass distributions, where instead of discussing the escape of discrete bodies, the escape of a certain mass fraction of the original body is considered. 
%\begin{corollary}
%Given a system modeled with a continuous distribution of mass $m$, its minimum energy potential is ${\cal U}_{min}(m)$. If $E < {\cal U}_{min}(m)\left[(1-\mu)^{5/3}+\mu^{5/3}\right]$, $0 < \mu < 1$, then a sub-component with mass fraction $\mu$ is Hill Stable. 
%\end{corollary}
%\begin{proof}
%The proof is a simple extension of that for Corollary 1. 
%\end{proof}

%It is of interest to directly compare the mutual potential of two different partitions, denoted here as ${\cal I}_1$ and ${\cal I}_2$. If ${\cal U}({\cal I}_1) < {\cal U}({\cal I}_2)$ then the energy threshold for partition 1 is lower and more likely to occur. For continuous mass distributions it is only necessary to show that the sum $S({\cal I}, 5/3) =  \sum_{i=1}^M \mu_i^{5/3}$ is greater for $1$ and $2$, or $S({\cal I}_1,5/3) > S({\cal I}_2,5/3)$. With this ranking, it becomes possible to directly compare different partitions in terms of the level of energy needed for them to lose their Hill Stability. 
%This ranking can also be done in terms of discrete bodies, however then one must evaluate all possible combinations using the fundamental integer partitions of the $N$ body problem. 

\section{Energy Classifications}

Given the basic conditions for the disruption of a partition,  the levels of energy that lead to these different results are now considered. 

\subsection{Failure Energy}

First consider the energy required to precipitate a change in the system's stability, or when it has a plastic failure and deformation that could result in fission. This ``failure energy'' can be split into two main components, one associated only with gravitational forces, denoted as $E_{G}$, and one associated with internal cohesive strength, denoted as $E_C$. Each of these are described below in more detail. The combination of these defines the failure energy, denoted here as $E_{Fail} = E_G + E_C$. 

\paragraph{Gravitational Failure Energy} This energy is a function of the current system configuration and overall angular momentum. For any resting system in an equilibrium configuration $\bfm{Q}^*$, there exists an angular momentum level at which the equilibrium ceases to exist or becomes unstable (assuming no cohesive strength). The gravitational failure energy $E_G$ then equals the value of the amended potential at this configuration and angular momentum, or $E_G(\bfm{Q}^*, H_G) = \frac{H_G^{2}}{2 I_H(\bfm{Q}^*)} + {\cal U}(\bfm{Q}^*)$. This transition has been studied for few-body systems in a number of publications \cite{scheeres_minE, F4BP_chapter, F3BP_scheeres}. For a given configuration, this will be the minimum energy needed to allow the system to change. It should be noted that for some configurations, the system is still Hill Stable even though a given equilibrium configuration may cease to exist. As this transition mode has been studied extensively analytically for discrete systems it is not described further. 

For a spherical continuous mass distribution a conservative upper limit on the spin rate for gravitational failure is the ``critical spin'' rate, $\omega_{G}$, equal to the orbit angular rate of a particle at the surface of the equal mass sphere, $\omega_G = \sqrt{{\cal G}m/R^3} = \sqrt{4\pi{\cal G}\rho/3}$, and is only a function of the body mean density and independent of size. In fact, this limiting spin rate is too large, as a cohesionless spherical body (with an angle of friction $< 90^\circ$) will fail at lower spin rates by landsliding or deformation \cite{holsapple_original, toshi_1950DA}.  
%It is interesting to note that a landsliding failure cannot lead to a direct escape of matter \cite{superfast}, although subsequent interactions with the rotating mass distribution may ultimately eject the body if the energy is high enough. 
Given the conservative nature of the bounded spin rate, an upper bound on the gravitational failure energy can be found in this case. Assuming a spherical body with moment of inertia $\frac{2}{5}m R^2$, a bound on the gravitational failure energy is found by combining the kinetic energy from rotation and the self-potential of the sphere. 
\beq
	E_G & < &  - \frac{2}{5}\frac{{\cal G}m^2}{R}
\eeq

\paragraph{Cohesive Failure Energy} 
%This energy component is always an additional energy that a configuration must be boosted to in order to precipitate failure. The system must first be at or beyond its gravitational failure energy. If it has cohesive strength it will not fail when it hits this threshold, as there is additional energy required to precipitate failure. This has been studied extensively numerically \cite{}, but only has relatively limited analytical analysis \cite{}. Here we only provide a general description of the required energy, derived for a local failure and a global failure. 

The effect of cohesion on the total energy budget must also be considered. This does not show up in the amended potential, but is an internal non-gravitational strength that must be overcome in order for the mass distribution to actually undergo a change in its configuration. This can be represented as an additional amount of work that the system must perform in order to precipitate break-up, fission or reconfiguration. Here we follow the model of granular cohesion presented in \cite{sanchez_MAPS}, which can be represented as a constant force that acts on a boulder emplaced in the regolith. 

If the cohesive strength is represented as $\sigma_c$, then the force that must be applied to extract a boulder of radius $r$ will be $2\pi r^2 \sigma_c$ and the distance it should be displaced is at least $r$. Thus, the total work done on the body will be on the order of $2\pi r^3\sigma_c$, and is proportional to the volume of the boulder and the cohesive strength. This can be equated to the additional energy which must be supplied in order to fission the body, or $ E_C \sim 2\pi r^3\sigma_c$. 

This analysis is also consistent with the global failure of a rubble pile as modeled in \cite{sanchez_MAPS}. There, a conservative estimate of the necessary spin rate for a body to undergo plastic deformation and fission was given as 
\beq
	\omega^2 & \sim  \frac{{\cal G}m}{R^3} + \frac{\sigma_Y}{\rho R^2}
\eeq
where $\sigma_Y$ is the yield strength, $\rho$ is the bulk density, $m$ is the body mass and $R$ is its mean radius.
The first component in this spin rate corresponds to the spin rate for gravitational failure, while the second term corresponds to the additional spin needed to break the cohesive bonds. 
Multiplying the second term by the moment of inertia and dividing by $2$ gives the additional cohesive energy $E_C$, 
\beq
	 E_C & \sim & \frac{1}{5}\frac{m }{\rho } \sigma_Y 
\eeq
which again is just a function of the total volume of the disrupted body times the yield strength. 

Ultimately, this additional energy $E_C$ must be given to the system for it to undergo failure, and leads to a higher spin rate (and hence higher energy) when it does fail. Once fission occurs, however, the cohesive forces will vanish once the bodies lose contact and the additional energy (in the form of kinetic energy) that was provided to drive the failure will be available for the system's dynamical evolution. 

\subsection{Free Energy}
Once a given configuration fails, i.e., its total energy equals $E_{Fail}$, not all of this energy is available for working on the system and causing sub-components to escape. Specifically, the self-energy of the components of the $N$-body problem cannot be liberated under the weak action of the gravitational forces alone. The free energy is defined as the failure energy that is available to work on the evolving components, and is what will eventually lead to any disruption and loss of Hill Stability. The self energy has previously been defined as the sum of all self-potentials, and now denote this as $E_{Self} = {\cal U}_{Self}$. This energy is not available for disrupting the system. Thus the free energy is $E_{Free} = E_{Fail} - E_{Self}$ and is the energy available for doing ``work'' on the system and allowing components to escape. Ultimately, this energy will be compared to the mutual potential of the sub-components to determine whether or not a partition is Hill Stable. The subtraction of the self energy is trivial for a discrete $N$-body system, but needs to be handled more carefully for continuous distributions of matter, which are discussed later. 

\subsection{Disassociation Energy}
Finally the disassociation energy for a given partition can be defined. This is an energy tied to the chosen partition of material and independent of the initial configuration and failure energy of the system. Here the disassociation energy is defined as the total minimum mutual potential energy of a given partition, which is the total minimum gravitational potential minus the self energy (i.e., self potentials).  
\beq
E_D({\cal I}) & = & \sum_{\mathbb{I}_i\in{\cal I}} {\cal U}_{min}({\mathbb{I}_i}) - {\cal U}_{Self} \\
& = & {\cal U}_{min}({\cal I}) - {\cal U}_{Self} 
\eeq

Given a system in equilibrium that is spun up to fission or reconfiguration the total energy of the system at failure can be represented as $E = E_G + E_C = E_{Fail} = {\cal E}(\bfm{Q}^*)$, where $\bfm{Q}^*$ is the configuration of the system at a relative equilibrium. If the system splits into a partition ${\cal I}$ that subsequently disrupts, then the amended potential goes to ${\cal E}\rightarrow{\cal U}({\cal Q})$ and maintains the inequality $E \ge {\cal U}({\cal Q}) \ge {\cal U}_{min}({\cal I})$. For the last step, one just needs to subtract the self potential energy from each side, $E_{Self} = {\cal U}_{Self}$, which for a system that has disrupted provides the necessary condition for disruption, $E_{Free} \ge E_D({\cal I})$. The converse of this gives the Hill Stability condition for the partition:
\beq
	E_{Free} < E_D({\cal I}) 
\eeq
This is the fundamental comparison for whether or not a given partition is Hill Stable. Ultimately, the free energy is a function of the initial equilibrium configuration and overall cohesive energy, while the disassociation energy is a function of the partition. These are then linked together through the Hill Stability theorems, which provide a framework for mapping observed Hill Unstable asteroid clusters to possible initial configurations and failure energies. 

\section{Minimum Potentials and Disassociation Energies}

For some special situations or under some restrictive assumptions the minimum gravitational potential of a set of bodies can be evaluated exactly, and hence the disassociation energy of certain partitions can be computed. These are outlined in the following. The discussion is divided into three parts, for discrete bodies, discrete bodies taken to the continuum limit and size distributions. 

%\subsection{Self Potential Energy} 
%Given matter with a constant density $\rho$ and a total mass $M$, the minimum gravitational potential of the material occurs when it is collected into a sphere, which will have a radius $R = \left[ \frac{3 M}{4\pi \rho} \right]^{1/3}$. The corresponding self gravitational potential energy is then ${\cal U}_{self} = - \frac{3 {\cal G} M^2}{5 R} = - \frac{3 {\cal G}}{5}\left(\frac{4\pi\rho}{3}\right)^2 R^5$, and thus scales as $R^5$. If the rigid body is distributed into an ellipsoidal shape the self energy can be computed as well in closed form, as discussed in \cite{spin_limits}, however in this paper we only use the spherical form of the self potential. 

\subsection{Equal Mass Spheres}
Consider a collection of spherical bodies of equal mass. Their minimum gravitational potential energy will in general occur when they are resting on each other in such a way that the overall maximum distance between any two of them is minimized. If given a minimum energy configuration for $N$ bodies, this principle can be used to recursively construct a candidate minimum energy configuration at $N+1$. 

Assume the spherical bodies have a diameter $d$ and mass $m$. First consider a single body, or $N = 1$, then the minimum potential energy (minus its self potential) is zero. For two bodies, the minimum distance the bodies can be from each other is a distance $d$, and the minimum potential ${\cal U}_{min}(2) = -\frac{{\cal G} m^2}{d}$, again neglecting the self-potentials. For $N=3$ a body can be added that touches both of the previous bodies, thus yielding a minimum potential of ${\cal U}_{min}(3) = - 3 \frac{{\cal G} m^2}{d}$. For $N=4$ the additional body can again touch all of the other bodies, forming a tetrahedron with a total of 6 mutual contacts and thus ${\cal U}_{min}(4) = -6 \frac{{\cal G} m^2}{d}$. 

The next body at $N=5$ can no longer be placed such that it touches all of the other bodies. More specifically, if placed on the surface, the new body can only touch up to 3 other spheres, and thus there are now other spheres that are some distance away. In this case the 5th body is placed on the face of the tetrahedron, touching 3 of the bodies with its distance to the other body being $2\sqrt{2/3}d \sim 1.633\ldots \ d$. This leads to ${\cal U}_{min}(5) = -\left[9+\frac{1}{2}\sqrt{\frac{3}{2}}\right] \frac{{\cal G} m^2}{d}$. {\color{black} We hypothesize that this is the minimum energy configuration for $N=5$, and that for up to $N=8$ this construction -- placing the next particle on the open face of the tetrahedron -- yields the minimum. This is not proven, however, and is currently only a supposition. }
At $N=6$, the sixth body is added on one of the remaining open faces of the $N=4$ tetrahedra, with the distance from the grain in the base tetrahedron it is not touching as above, and its distance to the fifth body placed being $5/3 d$. This can be continued similarly up to $N=8$, beyond which the computation becomes too complex for the current analysis. Table \ref{tab:min} lists these energies up to $N=8${\color{black}, with the proviso that the energies for $N = 5\rightarrow8$ are not proven to be minimum}. 

For higher numbers of $N$, weak lower and upper bounds on the mutual potential can be found. At the lower end one can assume that all $N$ bodies touch each other, while at the upper end it can be assumed that each new body only interacts with the 3 bodies it touches. The corresponding limits are
\beq
	-\frac{1}{2} N (N-1)\frac{{\cal G} m^2}{d} \le {\cal U}_{min}(N) \le - 3(N-2) \frac{{\cal G} m^2}{d}
\eeq
where the lower bound is exact for $1\le N \le 4$ and the upper bound is only valid for $N\ge3$, and is exact for $N = 3,4$. 
The limits are not very sharp, however, and could easily be improved. 

\begin{table}[h]
\centering
\caption{Minimum energy configurations as a function of $N$, along with lower and upper bounds. The energies are all normalized by the factor ${\cal G} m^2 / d$. }
\begin{tabular}{c | c | c | c }
$N$ & Lower & ${\cal U}_{min}(N)$ & Upper \\
\hline
\hline
2 & -1 & -1 & $-$ \\
\hline
3 & -3 & -3 & -3 \\
\hline
4 & -6 & -6 & -6 \\
\hline
5 & -10  & $-\left[ 9 + \frac{1}{2}\sqrt{\frac{3}{2}}\right] = -9.61\ldots$& -9 \\
\hline
6 & -15 & $-\left[ 12 + \frac{3}{5} + \sqrt{\frac{3}{2}}\right] = -13.82\ldots$  & -12 \\
\hline
7 & -21 & $-\left[ 16 + \frac{4}{5} + \frac{3}{2}\sqrt{\frac{3}{2}}\right] = -18.64\ldots$ & -15 \\
\hline
8 & -28 & $-\left[ 21 + \frac{3}{5} + \sqrt{6}\right] = -24.05\ldots$ & -18 \\
\hline
\end{tabular}
\label{tab:min}
\end{table}

To calculate the disassociation energies for different partitions of $N$ requires one to identify all of the integer partitions and explicitly enumerate their energies. For a given global partition ${\cal I} = \left\{\mathbb{I}_1, \mathbb{I}_2, \ldots, \mathbb{I}_M\right\}$, define $N_i = \mbox{card}(\mathbb{I}_i)$ as the number of bodies in the $i$th sub-partition. From Theorem 1 the disassociation energy is
\beq
	E_D({\cal I}) & = & \sum_{i=1}^M {\cal U}_{min}(N_i)
\eeq
where the self-potentials are not included. 
Figure \ref{fig:configurations} presents all integer partitions from $N=2 \rightarrow 8$, utilizing a Young Diagram. Also shown is the disassociation energy for each configuration. The Young diagram can be read as follows. For a given $N$, the lower left corner starts with the bodies all in a single group, and the boxed number next to that is the minimum energy for that grouping, taken from Table \ref{tab:min}. Moving up the diagram, the different colors represent the $N$ bodies being separated into different global partitions. For these, the number to the left of the partition is its disassociation energy, meaning that the total energy of the initial $N$ body system gathered into one body must be raised to that level as a necessary condition for the partition to be Hill Unstable. Moving up and then to the right one sees that the disassociation energies increase until at an energy of 0 the entire system can be disaggregated, meaning that it is possible for each body to mutually escape from each other. In accordance with Corollary \ref{cor:disc} note that the minimum energy of the configuration with $N$ bodies is the lowest disassociation energy of the configuration with $N+1$ bodies.  Thus, in principle all disassociation energies for the $N=9$ case could be computed, even without knowing the minimum energy of that number. Also note that for $N\ge6$ there can be multiple partitions with the same disassociation energy. 

\subsection{Mass distributions with $N\gg 1$}

For systems with $N\gg1$ it becomes impossible to track all of the possible computations and configurations. Here one is pushed to use a notation that explicitly separates the bodies according to mass fraction, and to carefully account for porosity and the possible packing density of a set of discrete grains. 
When packed together in a single body, the total mass must equal $m$, however they cannot be packed together with the same overall density due to geometric effects. Rather, the density of the packed body, called the bulk density, is in general equal to $\rho_b = (1-\phi) \rho_g$, where $\phi$ is the porosity and is a measure of how closely packed the grains can be. For an assemblage of equal size spheres the porosity can range from 0.26 to 0.48. 

Consider the case of $N$ bodies each of mass $m/N$, with the total mass of the system being $m$. Assume that each of these bodies has a ``grain density'' of $\rho_g$, such that $m/N=4\pi/3 \rho_g r^3$, where $r$ is the radius of each grain. Thus, the total volume of the packed body must be larger. Define the equivalent sphere of the same mass as  $m = 4\pi/3 \rho_b R^3$. Then relating this back to the grain density and number of bodies one finds $R^3 (1-\phi) = N r^3$, or $R = \left[ N / (1-\phi) \right]^{1/3} r$. 
The resulting minimum total potential energy of the body is then modeled for a sphere as
\beq
	{\cal U}_{min} & = & - \frac{3{\cal G} m^2}{5 R}
\eeq
where this energy includes both the self-potential energy of the individual grains and the mutual potential between different grains. 
In the following the implications of this are worked out for a system of $N \gg 1$ equal masses, and for arbitrary distributions of them. 

\subsubsection{Self-Potential for Equal Mass Spheres when $N\gg 1$}

When computing the minimum potential energy for such a body, one must be careful to account for the porosity and not count the self-potentials of the constituent grains. For a body of mass $m$ and radius $R$, the minimum energy configuration is to be in a sphere, with a total potential of 
\beq
	{\cal U}_{Mutual} + {\cal U}_{Self} & = & {\cal U}_{min} 
%	\\
%	{\cal U}_{min} & = & - \frac{3{\cal G} m^2}{5 R}
\eeq
If the body is made up of $N$ equal sized spheres of mass $m/N$ and radius $r$, and has an overall porosity of $\phi$, then a correction to the above can be computed. The self potential of each grain then is ${\cal U}_g = - \frac{3{\cal G}}{5} \frac{m^2}{N^2 \ r} $. Furthermore, the relationship between $N$, $r$ and $R$ can be approximated as $R = \left(\frac{N}{1-\phi}\right)^{1/3} r$. Thus, the total self-potential of a grain in terms of the bulk properties is ${\cal U}_g = - \frac{3{\cal G}}{5} \frac{m^2}{N^2} \frac{N^{1/3}}{ R(1-\phi)^{1/3}}$, and the total self potential is this times $N$
\beq
	{\cal U}_{Self} & = & \frac{1}{N^{2/3}(1-\phi)^{1/3}} {\cal U}_{min}
\eeq
Thus the minimum energy mutual potential of the body can be represented as 
\beq
	{\cal U}_{Mutual} & = & {\cal U}_{min} \left[ 1 -  \frac{1}{N^{2/3}(1-\phi)^{1/3}}\right]
\eeq
Dividing by the self potential of the original body, ${\cal U}_{min}$, defines the correction factor that accounts for the self-potentials.
\beq
	\overline{\cal U}_{Mutual} & = & \left[ 1 -  \frac{1}{N^{2/3}(1-\phi)^{1/3}}\right] \label{eq:mutual}
\eeq

Figure \ref{fig:self} shows the normalized minimum mutual energy as a function of $N$ for extreme values of porosity. From this we see that even for 1000 bodies the corrected energy is about 1\% larger than the total energy, and for 10 bodies it is around 25\% larger, a sizable fraction. However, as $N$ grows large enough the self-potential correction term vanishes. 
\begin{figure}[ht!]
\centering
\includegraphics[scale = 0.25]{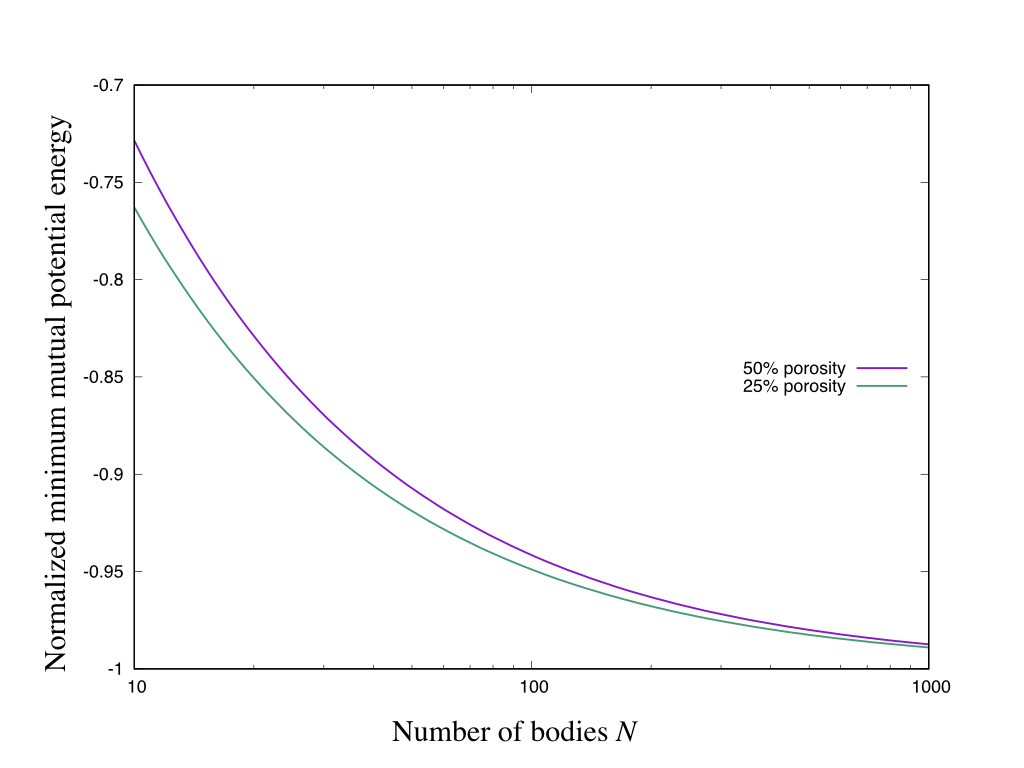}
\caption{Unit mutual potential energy as a function of $N$.}
\label{fig:self}
\end{figure}

\subsubsection{Disassociation Energies for $N\gg 1$}

Equation \ref{eq:mutual} accounts for the self-potentials of the discrete grains in the collection, however the collection of these grains themselves can be partitioned in many different ways. 
For the disassociation energy computation, we note that the minimum mutual potential with the self-potential correction is scaled by the overall minimum energy potential of the body and the same mass fraction scaling that is used in Corollary \ref{cor:cont} can be applied here. Thus, given a partition of the discrete mass elements ${\cal I}_M$, the disassociation energy can be computed as 
\beq
	E_D({\cal I}_M) & = & {\cal U}_{Mutual} \sum_{i\in {\cal I}_M} \mu_i^{5/3}
\eeq
When normalized by ${\cal U}_{min}$ this becomes 
\beq
	\overline{E}_D({\cal I}_M) & = & \sum_{i\in {\cal I}_M} \mu_i^{5/3} \left[ 1 -  \frac{1}{N^{2/3}(1-\phi)^{1/3}}\right]
\eeq
Here, is it important to distinguish between the number of partitions, $M$, and the total number of particles, $N$, where $M \le N$. 

It is instructive to work out a few specific examples for the disassociation energy. First, assume that the partition is $M$ equal mass components, so that each $\mu = 1/M$. Then the sum $\sum_{i\in {\cal I}} \mu_i^{5/3} = \frac{1}{M^{2/3}}$, and the energy required to disrupt the system will grow with $M$ as $E_D \propto - \frac{1}{M^{2/3}}$, and approaches zero as $M\rightarrow\infty$. 

If instead, the body is partitioned into a two components, the mass fractions will be $\mu$ and $(1-\mu)$ and the sum becomes $\sum_{i\in {\cal I}} \mu_i^{5/3} = \mu^{5/3} + (1-\mu)^{5/3}$. Here we see that the sum will be minimum at $\mu = 0.5$, equaling 0.63, and will take on a maximum equal to unity at $\mu = 0,1$. When converted to disassociation energy this is multiplied by $-1$, meaning that the equal mass system will have the largest disassociation energy and the infinitesimal mass the lowest. 

Finally, assume that the body is partitioned into a geometric series of the form $\mu_i = A \ w^i$, where $w < 1$ and $i = 1, 2, \ldots$. Here the tacit assumption is that $N\rightarrow \infty$ and thus that $M$ can be arbitrarily large. To satisfy the partition constraint it is required that $\sum_{i=1}^\infty A \ w^i = 1$. Recall the general result, $\sum_{i=1}^\infty w^i = \frac{w}{1-w}$. Thus the constraint requires that $A = \frac{1-w}{w}$, and thus the general form of the mass fraction is then $\mu_i = (1-w) \ w^{i-1}$. Then, the first and largest component of the system has a mass fraction $1-w$. The summation then becomes $\sum_{i\in {\cal I}} \mu_i^{5/3} = (1-w)^{5/3}\sum_{i=1}^\infty (w^{5/3})^{i-1} = \frac{(1-w)^{5/3}}{1-w^{5/3}}$. The limiting values of this function are unity for $w\rightarrow0$ and zero for $w\rightarrow 1$.

\subsubsection{Example Computations}
As an example, let us consider dividing a single body into a few different partitions and calculating their disassociation energies. We consider an initial body of a given mass and size, with the standard minimum potential ${\cal U}_{min} = -\frac{3{\cal G}m^2}{R}$ which we set to equal ${\cal U}_{min} = -1$. For definiteness we assume that $N\rightarrow\infty$, and thus that the total mutual potential energy correction factor is unity.  

Table \ref{tab:examples} shows the disassociation energy for a unit minimum potential for a number of different proposed partitions. These can all be viewed in the context of asteroid pairs, where we have systems of asteroids that have undergone some fission process and subsequently escaped from each other. The current population has been mostly viewed through the lens of a 2-body fission, however the current theory allows us to expand the possibilities. The table looks at a few energies for an $N=2$ disassociation, for $\mu = 0.5, 0.1, 0.01$. Here we clearly see that the needed energy for disassociation decreases as the mass fraction decreases to zero, and thus becomes easier to disassociate. Next the body is partitioned into $M$ equal masses, for $M = 2, 10, 100$. Here it can be noted that it takes considerably more energy to disassociate the body into a sequence of equal masses, with the necessary energy increasing with $M$. This implies that uniform disassociation of a body is more difficult than just losing a fraction of the total mass. Finally, a few infinite partitions of the body are considered, with mass ratios of the form $1/2^i$, $2 / 3^i$ and $9 / 10^i$. These systems have a decreasing disassociation energy with a decreasing $w$. 

\begin{table}[h]
\centering
\label{tab:examples}
\caption{Disassociation energies for a variety of configurations split into $M$ components. }
\begin{tabular}{ c | c | c }
$M$ & Mass Fractions & ${E}_D$ \\
\hline
\hline
2 & 0.5 & -0.63 \\
\hline
10 & $\mu = 0.1$ & -0.21 \\
\hline
100 & $\mu = 0.01$ & -0.05 \\
\hline
2 & 0.1, 0.9 & -0.86 \\
\hline
2 & 0.01, 0.99 & -0.98 \\
\hline
$\infty$ & $\mu_i = (1/2)^i$ & -0.46 \\
\hline
$\infty$ & $\mu_i = 2 \times (1/3)^i$ & -0.61 \\
\hline
$\infty$ & $\mu_i = 9 \times (1/10)^i$ & -0.86 \\
\hline
\end{tabular}
\end{table}

% (assuming that $M\rightarrow\infty$) such that $\mu_i = \frac{1}{2^i}$ for $i = 1, 2, \ldots$. Then from the classic summation formula for a geometric series, $\sum_{i=1}^\infty w^i = \frac{w}{1-w}$, it is easy to see that $\sum_{i=1}^\infty \frac{1}{2^i} = 1$, making this a valid partition.  Then $\sum_{i=1}^\infty \left(\frac{1}{2^{5/3}}\right)^i = \frac{1}{2^{5/3}-1} \sim 0.46\ldots$ and thus, $E_D = \frac{1}{2^{5/3}-1}{\cal U}_{min}$ for disruption into a geometric series. 

\subsection{Size Distributions}
Under some assumptions this theory can also be applied to the more general case of size distributions. The Appendix reviews size distributions for a general and two special cases which can be evaluated analytically. In general, define the cumulative size distribution by a function ${\cal N}(r)$ which gives the number of bodies in a collection of size greater than or equal to a radius $r$. The distribution is generally defined over an interval $r \in [r_0, r_1]$. Further, one can assume that there are a given number of boulders, ${\cal N}_1$, at the maximum size, or ${\cal N}(r_1) = {\cal N}_1$. Associated with the size distribution is the size distribution density, $n(r) = - d {\cal N}(r) / dr$. 

With such basic definitions the total mass of the system is 
\[
m = 4\pi/3 \rho_g \int_{r_0}^{r_1} r^3 \ n(r) \ dr
\]
where $\rho_g$ is the grain density. 
The same bulk density and porosity relationship can be assumed as in the previous section, so that $m = 4\pi/3 \rho_g (1-\phi) R^3$, which yields the equivalent spherical radius of the conglomeration as 
\[ R = \left[ \frac{1}{ (1-\phi)}\int_{r_0}^{r_1} r^3 \ n(r) \ dr   \right]^{1/3}
\]  
where $\phi$ is the porosity. 

The total self-potential for a distribution of grains will be 
\[ {\cal U}_{Self} = -\frac{3{\cal G}}{5} \left(\frac{4\pi\rho_g}{3}\right)^2 \int_{r_0}^{r_1}  r^5 n(r) \ dr
\] 
Then, following from the argument in the previous section it can be shown that 
\beq
	{\cal U}_{Mutual} + {\cal U}_{Self} & = & - \frac{3{\cal G}m^2}{5R} \\
	{\cal U}_{Mutual} & = & - \frac{3{\cal G}m^2}{5R} \left[ 1 - \frac{1}{R^5} \int_{r_0}^{r_1}  r^5 n(r) \ dr \right]
\eeq
which is also the minimum energy mutual potential. Then the normalized minimum mutual gravitational potential energy of the body is 
\beq
	\overline{\cal U}_{Mutual} & = & - \left[ 1 - \frac{1}{R^5} \int_{r_0}^{r_1}  r^5 n(r) \ dr \right]
\eeq
This will be evaluated for different size distributions and porosities. Note that the porosity appears implicitly only in the radius $R$, and that inserting the definition of that quantity the normalized energy becomes
\beq
	\overline{\cal U}_{Mutual} & = & - \left[ 1 - (1-\phi)^{5/3} \frac{\int_{r_0}^{r_1}  r^5 n(r) \ dr}{\left[\int_{r_0}^{r_1} r^3 \ n(r) \ dr\right]^{5/3}} \right]
\eeq
Here we note that, unlike the previous section, the self-potential correction does not vanish. 

Size distributions of the form ${\cal N}_{\alpha}(r) = {\cal N}_1 \left(\frac{r_1}{r}\right)^\alpha$, $2 \le \alpha \le 3$ are considered, which cover the range of size distributions observed for asteroidal bodies \cite{michikami, tsuchiyama_science}. In Table \ref{tab:dist} the total radius and normalized mutual potential for these different size distributions are summarized. The basic computations are covered in the Appendix. 

\begin{table}[h]
\centering
\caption{Radius and normalized minimum mutual potential for $2 \le \alpha \le 3$.}
\bigskip
\label{tab:dist}
\begin{tabular}{ c | c | c }
 $\alpha$ & $R$ & $\overline{\cal U}_{Mutual}$ \\
 \hline
\hline
& & \\
 2 & 
 $ r_1 \left[ \frac{2 {\cal N}_1 \left(1-\frac{r_0}{r_1}\right)}{1-\phi}\right]^{1/3}$ & 
 $ -\left[1-\frac{(1-\phi)^{5/3}}{3(2{\cal N}_1)^{2/3}} \frac{1-\left(\frac{r_0}{r_1}\right)^3}{\left[1-\frac{r_0}{r_1}\right]^{5/3} }\right]$ 
 \\
& &  \\
 \hline
& &  \\
 $2\le \alpha<3$ &  
 $ r_1 \left[ \frac{\alpha {\cal N}_1 \left[1-\left(\frac{r_0}{r_1}\right)^{3-\alpha}\right]}{(3-\alpha) (1-\phi)}\right]^{1/3}$ & 
 $ -\left[1-\frac{(1-\phi)^{5/3}(3-\alpha)^{5/3}}{(5-\alpha)(\alpha {\cal N}_1)^{2/3}} \frac{1-\left(\frac{r_0}{r_1}\right)^{5-\alpha}}{\left[1-\left(\frac{r_0}{r_1}\right)^{3-\alpha}\right]^{5/3} }\right]$  
 \\
& & \\
\hline
& &  \\
 3 & 
 $ r_1 \left[ \frac{3 {\cal N}_1 \ln\left(\frac{r_1}{r_0}\right)}{1-\phi}\right]^{1/3}$ & 
 $ -\left[1-\frac{(1-\phi)^{5/3}}{2(3 {\cal N}_1)^{2/3}} \frac{1-\left(\frac{r_0}{r_1}\right)^{2}}{\left[\ln\left(\frac{r_1}{r_0}\right)\right]^{5/3} }\right]$  
 \\
& & \\
\end{tabular}
\end{table}

In the ${\cal N}_\alpha$ distributions for $2 \le \alpha < 3$ the limit $r_0\rightarrow0$ can be taken without any singularities. For definiteness, take ${\cal N}_1 = 1$. This yields 
\beq
	\overline{\cal U}_{{Mutual}_\alpha} & = & -\left[1-\frac{(1-\phi)^{5/3}(3-\alpha)^{5/3}}{(5-\alpha)\alpha^{2/3}} \right]
\eeq
Figure \ref{fig:unorm} shows this potential as a function of size parameter $\alpha$ over the interval $2\le \alpha < 3$ and for different porosities. For low porosity one can note that the relative importance of the self potential rises, up to a 20\% increase in the minimum mutual potential at $\alpha=2$. For higher porosities, the relative contribution of the self potential becomes more muted, and at a porosity of 50\% leads to a $\sim7\%$ increase in minimum mutual potential at $\alpha=2$. The relative contributions all decrease for increasing $\alpha$. 

It is also important to consider how the total  mutual potential ${\cal U}_{min}$ also changes as a function of $\alpha$ and $\phi$. Again, taking $r_0 \rightarrow 0$ (for $\alpha<3$), setting ${\cal N}_1 = 1$, the scaling potential term becomes 
\beq
	{\cal U}_{min_\alpha} & = & - \frac{3 {\cal G} m^2}{5 r_1} \left[ \frac{(3-\alpha)(1-\phi)}{\alpha}\right]^{1/3}
\eeq
Figure \ref{fig:mutual} shows this as a function of $\alpha$ and $\phi$, scaled by the quantity $\frac{3 {\cal G} m^2}{5 r_1}$. Here there is a large change as a function of $\alpha$, with the steeper size distributions (i.e., larger $\alpha$) having a much higher overall energy. This means that, for a given fixed total mass $m$, the rubble piles with a steeper size distribution will be more difficult to disrupt. 

At $\alpha=3$ the minimum mutual potential energy takes on a different form, and one cannot let $r_0\rightarrow0$ in general. Instead, again let ${\cal N}_1 = 1$ but now consider a size scale of $r_1 / r_0 \sim 10^3 \rightarrow 10^6$, which bounds what was confirmed on Itokawa \cite{michikami, tsuchiyama_science} and assures that the correction term is negligible. 
The normalized and scaling potentials are then 
\beq
	\overline{\cal U}_{{Mutual}_3} & = & -\left[1-\frac{(1-\phi)^{5/3}}{2 \ 3^{2/3}} \left[\ln\left(\frac{r_1}{r_0}\right)\right]^{-5/3}\right] \\
	\overline{\cal U}_{min_3} & = & - \left[ \frac{1-\phi}{3  \ln\left(\frac{r_1}{r_0}\right)} \right]^{1/3}
\eeq
and the $\ln(r_1/r_0)$ term then takes on values from $6.9\rightarrow 13.8$, and thus can be approximated as $\sim 10$. 
For a porosity of $\phi = 0$ the relative potential is $\sim -0.995$ and the scaling potential is -0.32. Increasing the porosity to 50\% gives the relative potential to be $\sim -0.998$ and the scaling potential to be -0.25. Thus, at $\alpha=3$ the effect of the self potential on the overall minimum mutual potential is relatively small, however the scaling potential has significant variations due to porosity. 

 \begin{figure}[ht!]
\centering
\includegraphics[scale = 0.25]{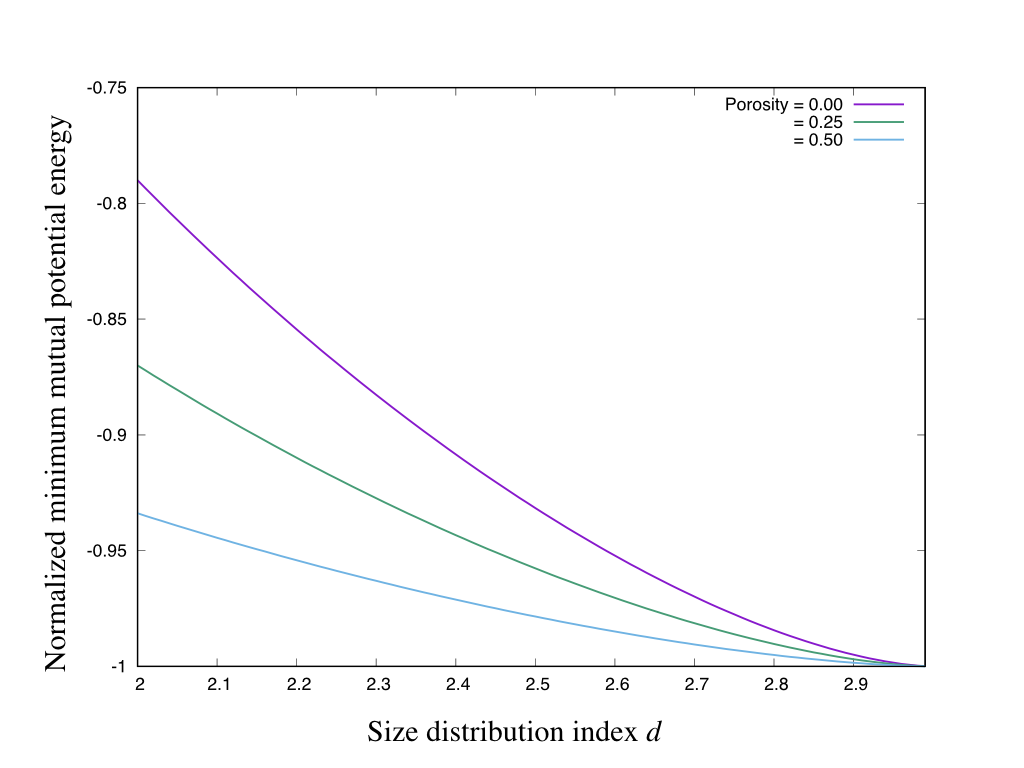}
\caption{Minimum mutual potential energy for size distributions as a function of $\alpha$ over the interval $[2,3)$.}
\label{fig:unorm}
\end{figure}

 \begin{figure}[ht!]
\centering
\includegraphics[scale = 0.25]{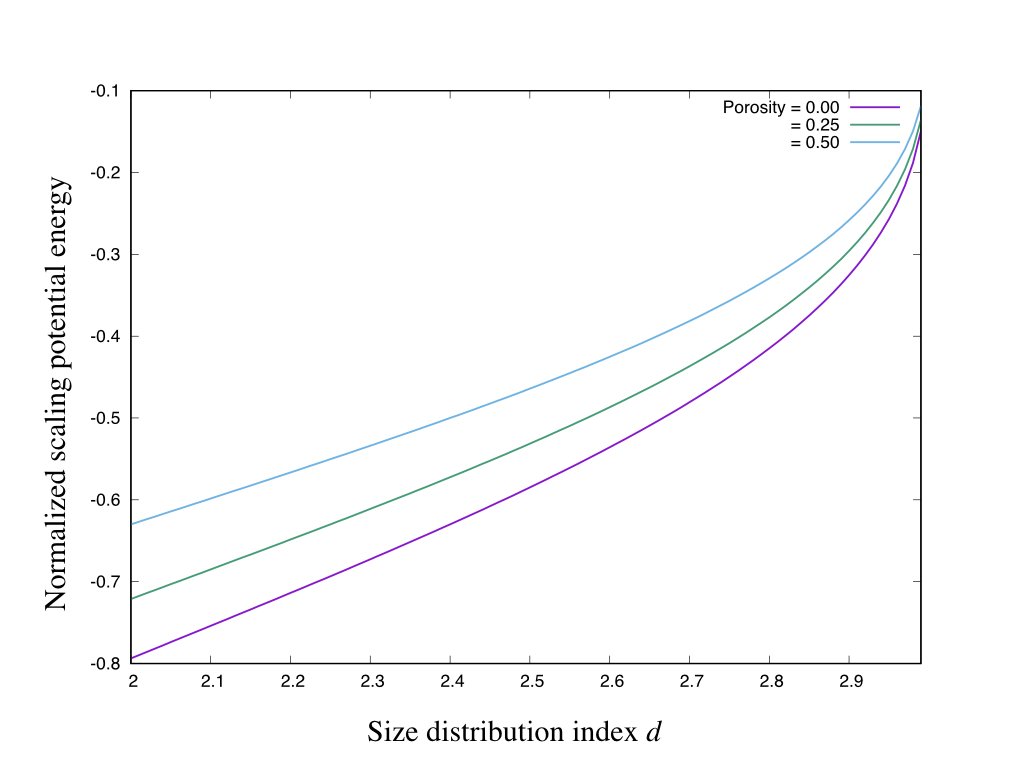}
\caption{Scaling potential energy for size distributions as a function of $\alpha$ over the interval $[2,3)$.}
\label{fig:mutual}
\end{figure}

It is important to note that the minimum energy potentials for a size distribution are independent of the total number of grains, and thus they give a systematic correction to the disassociation energies. Following from the previous section, for a given partition of the original body, ${\cal I}$, the associated disassociation energy is now
\beq
	E_D({\cal I}) & = & {\cal U}_{{Mutual}_\alpha} \sum_{i\in{\cal I}} \mu_i^{5/3}
\eeq

This formula ignores the delicate issue of splitting size distributions into smaller mass fractions. For a given distribution with ${\cal N}_1$ boulders of size $r_1$, it is an open question how these should be divided if the size distribution is divided into sub-partitions. To be more precise, it would be necessary to identify the largest boulder within any given sub-partition, and then reconstitute the size distribution for that sub-component. However, unless the largest boulders are divided up proportionate to the total mass divisions, this will cause some of the size distributions to deviate from the ideal forms given here. We recognize and identify this as an issue here, but do not pursue it any further. 
Conveniently ignoring this aspect, these mutual potentials can then be used to compute disassociation energies as outlined above. 

An important point should be made comparing the disassociation energies of the different size distributions. If two parent bodies with the same mass and radius but different porosities are compared, the relative disassociation energies are defined in Fig.\ \ref{fig:unorm}. As the porosity is increased and as the size distribution is made steeper, the energy decreases, meaning that a more porous body with a steeper size distribution will be relatively easier to disaggregate into any given partition. If instead, the total mass of the asteroid is fixed along with its largest boulder, $r_1$, the trend is different for the overall scaling mutual potential, as shown in Fig.\ \ref{fig:mutual}, which shows an increase in energy with increasing porosity and with steeper size distributions, however this does not account for the changes shown in Fig.\ \ref{fig:unorm}. Multiplying these two together provides a better handle on the total mutual potential, shown in Fig.\ \ref{fig:joint}. The direct comparison is for asteroids of the same mass and with the same largest size boulder, however the total sizes of these bodies may change drastically as a function of size distribution and porosity. Here it is clear that the overall trend is for the disassociation energy to increase with increasing porosity and steepness. 

 \begin{figure}[ht!]
\centering
\includegraphics[scale = 0.25]{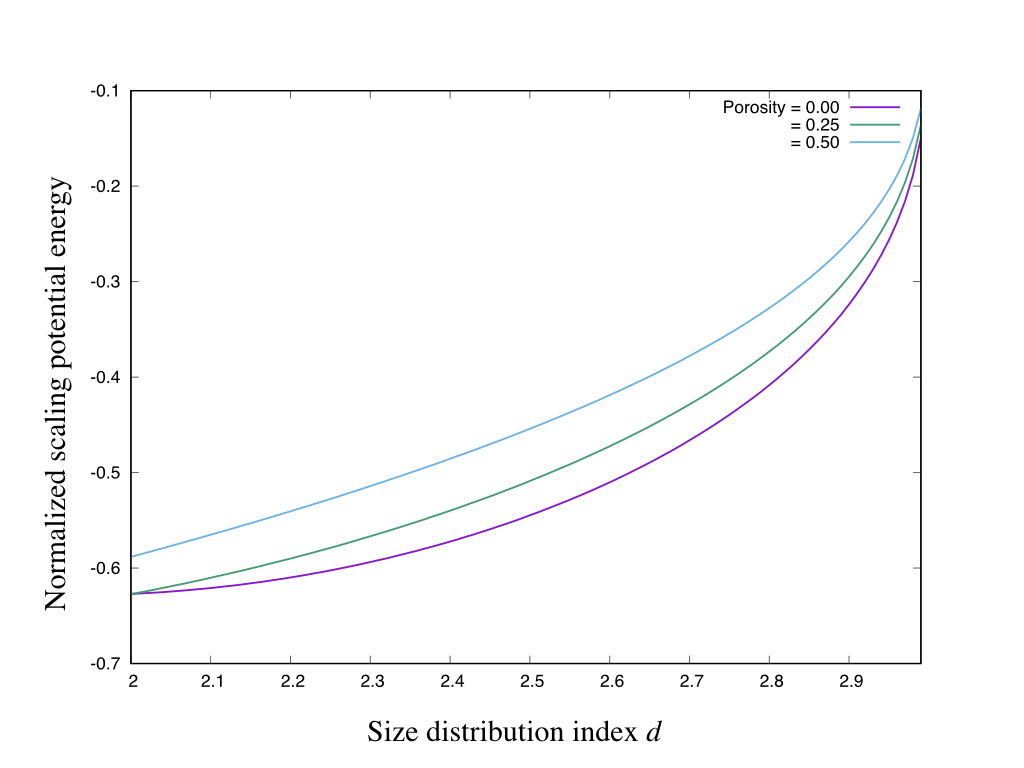}
\caption{Scaling potential energy times normalized mutual potential for size distributions as a function of $\alpha$ over the interval $[2,3)$ and for different porosities.}
\label{fig:joint}
\end{figure}

\section{Conclusions}
This paper presents a series of definitions and basic theorems that can be used to calculate the disassociation energy of an $N$ body problem into different sets of asteroid clusters. These results are developed for the discrete mass problem, the continuum limit of this problem, and for the more general case of size distributions. Resulting are rigorous results for the level of energy required to produce asteroid clusters from rubble pile asteroids.

\section*{Appendix: Size Distribution Functions}

Consider a cumulative size distribution of the form ${\cal N}_{\alpha}(r) = \frac{A_\alpha}{r^\alpha}$ for $2 \le \alpha \le 3$. Associated with this distribution is a maximum and minimum grain radius, $r_1$ and $r_0$, respectively. 
The function ${\cal N}_{\alpha}(r)$ is the cumulative number of particles with radius between $r$ and the maximum size $r_1$. The term $A_\alpha$ is initially chosen to agree with the observed number of largest boulders, ${\cal N}_1$, such that ${\cal N}_{\alpha}(r_1)={\cal N}_1$. With this interpretation, the nominal form for the function is:
\beq
	{\cal N}_{\alpha}(r) & = & {\cal N}_1 \left(\frac{r_1}{r}\right)^\alpha
\eeq

The cumulative distribution is the integral of a cumulative density function $n_\alpha(r)$, defined as:
\beq
	{\cal N}_{\alpha}(r) & = & \int_{r}^{r_1} n_\alpha(r) \ dr
\eeq
This definition establishes that $n_\alpha(r) = - \frac{d{\cal N}_{\alpha}}{dr}$, leading to the cumulative density function
\beq
	n_\alpha(r) & = & \frac{\alpha {\cal N}_1 \ r_1^\alpha}{r^{\alpha+1}}
\eeq
A density distribution function that integrates to unity can also be defined, denoted as $\bar{n}_\alpha(r)$:
\beq
	\bar{n}_\alpha(r) & = & \frac{n_\alpha(r)}{\int_{r_0}^{r_1} n_\alpha(r) \ dr}
\eeq
Carrying out this computation yields 
\beq
	\bar{n}_\alpha(r) & = & \frac{\alpha r_1^\alpha r_0^\alpha}{(r_1^\alpha - r_0^\alpha) r^{\alpha+1}} 
\eeq

There are several quantities of interest that can be defined and calculated with a power law size distribution. A few of them are reviewed here, in addition to stating some key results used in the paper. 

\paragraph{Mean Grain Radius}
The mean grain radius is defined as
\beq
	\bar{r} & = & \int_{r_0}^{r_1} r \bar{n}_\alpha(r) \ dr \\
	& = &  \frac{\alpha r_1 r_0}{\alpha-1}\frac{ r_1^{\alpha-1} - r_0^{\alpha-1}}{r_1^\alpha - r_0^\alpha}  
\eeq
Thus if $r_0 \ll r_1$ the mean radius is $\bar{r} \sim \frac{\alpha}{\alpha-1} r_0$. 

\paragraph{Surface Area of Grains}

The total surface area of a collection of grains is computed as
\beq
	{SA}_T & = & \int_{r_0}^{r_1} 4\pi r^2 {n_\alpha}(r) \ dr \\
	& = & 4\pi {\cal N}_1 \alpha r_1^\alpha \int_{r_0}^{r_1} r^{1-\alpha} dr
\eeq
If $2 < \alpha \le 3$ this can be integrated to find
\beq
	{SA}_T & = & \frac{4\pi {\cal N}_1 \alpha }{\alpha-2} r_1^2 \left[ \left(\frac{r_1}{r_0}\right)^{\alpha-2} - 1\right]   
\eeq
and if $\alpha = 2$ the total surface area equals
\beq
	{SA}_T & = & 8\pi {\cal N}_1 r_1^2 \ln\left(\frac{r_1}{r_0}\right)
\eeq
For either case, if $r_0 \ll r_1$, the total surface area becomes arbitrarily large.

\paragraph{ Volume of Grains}

The total volume of grains can be found by 
\beq
	{V}_T & = & \int_{r_0}^{r_1} \frac{4\pi}{3} r^3 {n_\alpha}(r) \ dr \\
	& = & \frac{4\pi \ \alpha {\cal N}_1 }{3} r_1^\alpha  \int_{r_0}^{r_1}  r^{2-\alpha} \ dr
\eeq
If $2 \le \alpha < 3$ the total volume equals
\beq
		{V}_T & = & \frac{4\pi}{3} \frac{\alpha {\cal N}_1 r_1^3}{3-\alpha} \left[1 - \left(\frac{r_0}{r_1}\right)^{3-\alpha}\right]
\eeq
If $\alpha = 3$ the total volume equals 
\beq
	{V}_T & = & 4\pi {\cal N}_1 r_1^3 \ln\left(\frac{r_1}{r_0}\right)
\eeq
For $\alpha < 3$ one can take the limit $r_0\rightarrow\infty$ without any singularity. For $\alpha=3$, however, this leads to an infinite mass. 

\paragraph{Total Self Potential}

Finally, the total self potential of a size distribution, assuming spherical grains, is computed as
\beq
	{\cal U}_{Self} & = & -\frac{3{\cal G}}{5} \left(\frac{4\pi\rho_g}{3}\right)^2 \int_{r_0}^{r_1}  r^5 n_\alpha(r) \ dr \\
	& = & -\frac{3{\cal G}}{5} \left(\frac{4\pi\rho_g}{3}\right)^2 {\alpha {\cal N}_1 r_1^\alpha} \int_{r_0}^{r_1}  r^{4-\alpha} \ dr
\eeq
The integral is defined for the whole interval of $2\le \alpha \le 3$, yielding 
\beq
	{\cal U}_{Self} & = & -\frac{3{\cal G}}{5} \left(\frac{4\pi\rho_g}{3}\right)^2 \frac{\alpha {\cal N}_1 r_1^5}{5-\alpha} \left[ 1 - \left(\frac{r_0}{r_1}\right)^{5-\alpha}\right] 
\eeq
Across the entire interval the limit $r_0\rightarrow0$ can be taken.

\paragraph{Acknowledgements} { The author appreciates the comments of the two reviewers, which have helped to greatly improve this paper.}

\paragraph{Conflict of Interest Statement} { The author declares that no conflict of interest exists with the research reported herein.}

%\bibliographystyle{plain}
%\bibliography{../../bibliographies/biblio_conferences,../../bibliographies/biblio_article,../../bibliographies/biblio_books,../../bibliographies/biblio_thesis,../../bibliographies/biblio_misc}

\end{document}